\documentclass[a4paper,UKenglish,cleveref, autoref, thm-restate]{lipics-v2021}
\nolinenumbers

\usepackage{amsbsy, amsmath, amscd, amssymb, pifont, latexsym}

\usepackage{mathrsfs}

\usepackage[utf8]{inputenc}
\usepackage{latexsym}
\usepackage{graphicx}
\usepackage{epsf}
\usepackage{color}
\usepackage{enumerate}
\usepackage{tabularx}
\usepackage{amsmath}
\usepackage{amssymb}
\usepackage{amstext}
\usepackage{float}
\usepackage{hyperref}
\usepackage{tikz}
\usepackage[bibliography=common]{apxproof}

\usepackage[T1]{fontenc}

\newcommand{\makeprob}[3]{
\medskip\noindent
\begin{tabularx}{\textwidth}{@{}l@{ }>{\raggedright\arraybackslash}X}
\multicolumn{2}{@{}l}{#1}\\
\textsf{Input: } & {#2} \\
\textsf{Output: }& {#3} \\
\end{tabularx}
}

\newcommand{\pb}[1]{\textsc{#1}}

\newcommand{\calB}{{\cal B}}
\newcommand{\calC}{{\cal C}}

\newcommand{\calD}{{\cal D}}
\newcommand{\calE}{{\cal E}}

\newcommand{\calK}{{\cal K}}

\newcommand{\calM}{{\cal M}}

\newcommand{\calP}{{\cal P}}

\newcommand{\N}{{\mathbb{N}}}

\newcommand{\raus}[1]{}

\newcommand{\Etransversal}{\text{\sc Enum-Transversal}}
\newcommand{\Emonotonesat}{\text{\sc Enum-Monotone-Sat}}

\newcommand{\Eihssat}{\text{\sc Enum-IHS-Sat}}
\newcommand{\Exorsat}{\text{\sc Enum-XOR-Sat}}
\newcommand{\Ereach}{\text{\sc Enum-Reach}}
\newcommand{\tc}{\text{tc}}
\newcommand{\EHornsat}{\text{\sc Enum-Horn-Sat}}
\newcommand{\EKromsat}{\text{\sc Enum-Krom-Sat}}
\newcommand{\EGrayWord}{\text{\sc Enum-Gray-Word}}
\newcommand{\EGrayRank}{\text{\sc Enum-Gray-Rank}}
\newcommand{\EGrayWordord}{\mbox{\text{\sc Enum-Gray-Word}$_{ord}$}}
\newcommand{\EGrayRankord}{\mbox{\text{\sc Enum-Gray-Rank}$_{ord}$}}
\newcommand{\enum}[1]{\textsc{Enum}\smash{\cdot}#1}

\newcommand{\sol}[2]{\mathrm{Sol}_{#1}({#2})}

\newcommand{\ACzero}{\mathbf{AC}^0}
\newcommand{\NCone}{\mathbf{NC}^1}

\newcommand{\NP}{\mathbf{NP}}
\newcommand{\PSPACE}{\mathbf{PSPACE}}

\newcommand{\ptime}{\mathbf{P}}

\newcommand{\cdl}{\mathbf{CD}\!\circ\!\mathbf{lin}}

\newcommand{\IncP}{\mathbf{IncP}}
\newcommand{\DelayP}{\mathbf{DelayP}}

\newcommand{\DelACzero}{\mathbf{Del}\!\cdot\!\mathbf{AC}^0}
\newcommand{\DelK}{\mathbf{Del}\!\cdot\!\calK}
\newcommand{\DelACzeroPrecomp}[1]{\mathbf{Del}_{#1}\!\cdot\!\mathbf{AC}^0}
\newcommand{\DelKPrecomp}[1]{\mathbf{Del}_{#1}\!\cdot\!\calK}
\newcommand{\DelACzeroPrecompSpace}[1]{\mathbf{Del}^*_{#1}\!\cdot\!\mathbf{AC}^0}
\newcommand{\DelKPrecompSpace}[1]{\mathbf{Del}^*_{#1}\!\cdot\!\calK}
\newcommand{\DelACzeroPrecompPolySpace}[1]{\mathbf{Del}^{\mathrm{P}}_{#1}\!\cdot\!\mathbf{AC}^0}
\newcommand{\DelKPrecompPolySpace}[1]{\mathbf{Del}^{\mathrm{P}}_{#1}\!\cdot\!\calK}
\newcommand{\DelACzeroPrecompConstantSpace}[1]{\mathbf{Del}^{\mathrm{c}}_{#1}\!\cdot\!\mathbf{AC}^0}
\newcommand{\DelKPrecompConstantSpace}[1]{\mathbf{Del}^{\mathrm{c}}_{#1}\!\cdot\!\calK}
\newcommand{\DelACzeroSpace}{\mathbf{Del}^*\!\cdot\!\mathbf{AC}^0}
\newcommand{\DelKSpace}{\mathbf{Del}^*\!\cdot\!\calK}
\newcommand{\DelACzeroPolySpace}{\mathbf{Del}^{\mathrm{P}}\!\cdot\!\mathbf{AC}^0}
\newcommand{\DelKPolySpace}{\mathbf{Del}^{\mathrm{P}}\!\cdot\!\calK}

\newcommand{\DelACzeroConstantSpace}{\mathbf{Del}^{\mathrm{c}}\!\cdot\!\mathbf{AC}^0}
\newcommand{\DelKConstantSpace}{\mathbf{Del}^{\mathrm{c}}\!\cdot\!\calK}

\title{Enumeration Classes Defined by Circuits}

 \author{Nadia Creignou}{Aix Marseille Univ, Université de Toulon, CNRS, LIS, Marseille, France}{Nadia.Creignou@lis-lab.fr}{[orcid]}{}
 
  \author{Arnaud Durand}{Université Paris Cité, CNRS, IMJ-PRG, Paris, France}{durand@math.univ-paris-diderot.fr}{[orcid]}{}
 
 \author{Heribert Vollmer}{Leibniz Universität Hannover}{vollmer@thi.uni-hannover.de}{https://orcid.org/0000-0002-9292-1960}{}
 
\authorrunning{N. Creignou, A.Durand and H. Vollmer} 

\Copyright{Nadia Creignou, Arnaud Durand and Heribert Vollmer} 



\begin{document}

\maketitle

%

\begin{abstract}
We refine the complexity landscape for enumeration problems by introducing very low classes defined by using Boolean circuits as enumerators. We locate well-known enumeration problems, e.g., from graph theory, Gray code enumeration, and propositional satisfiability in our classes. 
In this way we obtain a framework to distinguish between the complexity of different problems known to be in $\DelayP$, for which a formal way of comparison was not possible to this day.
\end{abstract}

\section{Introduction}

In computational complexity theory, most often decision problems are studied that ask for the existence of a solution to some problem instance, e.g., a satisfying assignment of a given propositional formula. In contrast, enumeration problems ask for a list of all solutions, e.g., all satisfying assignments. In many application areas these are the more ``natural'' kind of problems---let us just mention database queries, web search, diagnosis, data mining, bioinformatics, etc.

The notion of \emph{tractability} for enumeration problems requires a new approach, simply because there may be a large number of solutions, exponential in the input size. Widely studied is the class $\DelayP$ ("polynomial delay"), containing all enumeration problems where, for a given instance $x$, (i) the time to compute the first solution, (ii) the time between producing any two consecutive solutions, and (iii) the time to detect that no further solution exists, are all polynomially bounded in the length of $x$. Also the class $\IncP$ ("incremental polynomial time"), where we allow the time to produce the next solution and to signal that no further solution exists to grow by a polynomial bounded in the size of the input \emph{plus} the number of already computed solutions. These classes were introduced in 1988 in \cite{JYP88}, and since then, an immense number of membership results have been obtained. Recently, also intractable enumeration problems have received some attention. Reducibilities, a completeness notion and a hierarchy of intractable enumeration problems, analogous to the well-known polynomial hierarchy, were defined and studied in \cite{CKPSV19}.

In this paper we will look for notions of tractability for enumeration stricter than the above two. More specifically, we will introduce a refinement of the existing classes based on the computation model of Boolean circuits. The main new class in our framework is the class $\DelACzero$. An enumeration problem belongs to this class if there is a family of $\ACzero$ circuits, i.e., a family of Boolean circuits of constant depth and polynomial size with unbounded fan-in gates, that (i) given the input computes the first solution, (ii) given input and a solution computes the next solution (in any fixed order of solutions), and (iii) given input and the last solution, signals that no further solution exists.
Still using $\ACzero$ circuits 
we then consider extended classes by allowing
\begin{itemize}
    \item precomputation of different complexity (typically, polynomial time precomputation) and/or
    \item memory to be passed on from the computation of one solution to the next (from a constant to a polynomial number of bits) 
\end{itemize}
By this, we obtain a hierarchy of classes within $\DelayP$/$\IncP$ shown in Fig.~\ref{graphe1}.

The main motivation behind our work is the wish to be able to compare the complexity of different tractable enumeration problems by classifying them in a fine hierarchy within $\DelayP$, and to obtain lower bounds for enumeration tasks. From different application areas such as graph problems, Gray code enumeration and satisfiability, we identify natural problems, all belonging to $\DelayP$, some of which can be enumerated in $\DelACzero$, some cannot, but allowing precomputation or a certain number of bits of auxiliary memory they can. We would like to mention in particular the maybe algorithmically most interesting contribution of our paper, the case of enumeration for satisfiability of 2-CNF (Krom) formulas. While it is known that counting satisfying assignments for formulas from this fragment of propositional logic is $\#\ptime$-complete \cite{Valiant79enum}, we exhibit a $\DelACzeroPrecomp{\ptime}$ algorithm (i.e. $\DelACzero$ with polynomial time precomputation but no memory), for enumeration, thus placing the problem in one of the lowest class in our framework. 
This means that 
surprisingly satisfying assignments of Krom formulas can be enumerated very efficiently (only $\ACzero$ is needed to produce the next solution) after a polynomial time precomputation before producing the first solution.

Building on well-known lower bounds (in particular for the parity function \cite{FSS84,Ajtai89})
we prove (unconditional) separations among (some of) our classes and strict containment in $\DelayP$, and building on well-known completeness results we obtain conditional separations, leading to the inclusions and non-inclusions depicted in Fig.~\ref{graphe1}.

Another refinement of $\DelayP$ that has received considerable attention in the past, in particular in the database community, is the class $\cdl$ of problems that can be enumerated on RAMs with constant delay after linear time preprocessing \cite{DurandG07} (see also the surveys~\cite{DBLP:conf/stacs/Segoufin14,DBLP:conf/pods/durand}). 
It is not difficult to see (see Section~\ref{subsec:relation to known classes}) that $\cdl$ and $\DelACzero$ are incomparable classes; thus our approach provides a novel way to refine polynomial delay.

This paper is organized as follows. After some preliminaries, we introduce our new classes in Sect.~\ref{definitions}. In Sect.~\ref{examples} we present a number of upper and lower bounds for example enumeration problems from graph theory, Gray code enumeration and propositional satisfiability.  Depending whether we allow or disallow precomputation steps, we obtain further conditional or unconditional separation results between classes in Sect.~\ref{separations}. Finally we conclude with a number of open problems.

\section{Preliminaries}

Since our main computational model will be Boolean circuits, we fix the alphabet $\Sigma=\{0,1\}$, and use this alphabet to encode graphs, formulas, etc., as usual. Any reasonable encoding will do for all of our results. 

Let $R\subseteq\Sigma^*\times\Sigma^*$ be a computable predicate. We say that $R$ is polynomially balanced, if there is a polynomial $p$ such that for all pairs $(x,y)\in R$, we have $|y|\leq p(|x|)$. 
Now we define the enumeration problem associated to $R$ as follows.

\makeprob{$\enum{R}$}{$x\in \Sigma^*$}{an enumeration of elements in $\sol{R}{x}=\{y: R(x,y)\}$}


We require that $R$ is computable but do not make any complexity assumptions on $R$. In the enumeration context, it is sometimes stipulated that $R$ is polynomial-time checkable, i.e., membership of $(x,y)$ in $R$ is decidable in time polynomial in the length of the pair \cite{Str19,CapelliS19}. Generally, we do not require this, but we will come back to this point later.

We assume basic familiarity of the reader with the model of Boolean circuits, see, e.g., \cite{Vol99,CK02}. We use $\ACzero$ to denote the class languages that can be decided by uniform families of Boolean circuits of polynomial size and constant depth with gates of unbounded fan-in. The class of functions computed by such circuit families is denoted by $\mathbf{F}\ACzero$, and for simplicity often again by $\ACzero$.
The notation for the corresponding class of languages/functions defined by uniform families of circuits of polynomial size and logarithmic depth with gates of bounded fan-in is $\NCone$. 

The actual type of uniformity used is of no importance for the results of the present paper.
However, for concreteness, all circuit classes in this paper are assumed to be uniform using the ``standard'' uniformity condition, i.\,e., DLOGTIME-uniformity/U$_\textrm{E}$-uniformity \cite{BIS90}; the interested reader may also consult the textbook \cite{Vol99}.



\section{Delay Classes with Circuit Generators}
\label{definitions}

In this section we present the formal definition of our new enumeration classes. As we already said, we will restrict our definition to usual delay classes; classes with incremental delay can be defined analogously, however, we will see that our delay-classes with memory in a sense reflect incremental classes in the circuit model.

The main idea is that the generation of a \emph{next} solution will be done by a circuit from a family; in the examples and lower and upper bounds in the upcoming sections, these families are usually of low complexity like $\ACzero$ or $\NCone$. The generator will receive the original input word plus the previous solution. Parameters in the definition will be first the complexity of any precomputation before the first solution is output, and second the amount of information passed from the generation of one solution to the next.


\subsection{Delay Classes with no Memory}

For a family $\calC=(C_n)_{n\in\N}$ of Boolean circuits, circuit $C_i$ will be the circuit in the family with $i$ input gates. When the length of the circuit input is clear from the context, we will usually simply write $C_{|\cdot|}$ to refer to the circuit with appropriate number of input gates.

\begin{definition}\label{def:delAC0} 
{\rm [$\calK$-delay] }

Let $R$ be a polynomially balanced predicate. The enumeration problem $\enum{R}$ is in $\DelK$ if there exists a 
family of $\calK$-circuits $\calC=(C_n)_{n\in\N}$ such that, for all inputs $x$, there is an enumeration $y_1,...,y_k$ of $\sol{R}{x}$ and:

\begin{itemize}
\item $C_{|\cdot|}(x)=y_1\in \sol{R}{x}$, 
\item for all $i<k$: $C_{|\cdot|}(x,y_i)=y_{i+1}\in \sol{R}{x}$
\item $C_{|\cdot|}(x,y_k)=y_k$ 
\end{itemize}
  \end{definition}
  
  
Note that by the last requirement, the circuit family signals there is no further solution if the input solution is given again as output.
Moreover, we point out that, in the definition above, if $x$ is an input and $y\in \sol{R}{x}$, then $C_{|x|+|y|}$ produces a $z\in \sol{R}{x}$. However, if $y\not\in \sol{R}{x}$, nothing is specified about the output $z$. 


Next we consider classes where a precomputation before outputting the first solution is allowed. The ressource bounds of the precomputation are specified by an arbitrary complexity class.

\begin{definition}\label{def:DelAC0Precomp} 
{\rm [$\calK$-delay with $T$-precomputation] }
 
Let $R$ be a polynomially balanced predicate and $T$ be a complexity class. The enumeration problem $\enum{R}$ is in $\DelKPrecomp{T}$  if there exists an algorithm $M$ working with resource $T$ and a  family of $\calK$-circuits $\calC=(C_n)_{n\in\N}$ such that, for all input $x$ there is an enumeration $y_1,...,y_k$ of $\sol{R}{x}$ and:

\begin{itemize}
\item $M$ compute some value $x^*$, i.e., $M(x)=x^*$
\item $C_{|\cdot|}(x^*)=y_1\in \sol{R}{x}$, 
\item for all $i<k$: $C_{|\cdot|}(x^*,y_i)=y_{i+1}\in \sol{R}{x}$
\item $C_{|\cdot|}(x^*,y_k)=y_k$ 
\end{itemize}
\end{definition}

\subsection{Delay Classes with memory}

Extending the above model, we now allow each circuit to produce slightly more than the next solution. These additional information is then passed as extra input to the computation of the next solution, in other words, it can serve as an auxiliary memory.

\begin{definition}\label{def:DelAC0Space} {\rm [$\calK$-delay with auxiliary memory] }

Let $R$ be a polynomially balanced predicate. The enumeration problem $\enum{R}$ is in $\DelKSpace$  if there exist two families of $\calK$-circuits $\calC=(C_n)_{n\in\N}$, $\calD=(D_n)_{n\in\N}$ such that, for all input $x$ there is an enumeration $y_1,...,y_k$ of $\sol{R}{x}$ and:

\begin{itemize}
\item $C_{|\cdot|}(x)=y^*_1$ and $D_{|\cdot|}(y^*_1)=y_1\in \sol{R}{x}$, 
\item for all $i<k$: $C_{|\cdot|}(x,y^*_i)=y^*_{i+1}$ and $D_{|\cdot|}(y^*_{i+1})=y_{i+1}\in \sol{R}{x}$, 
\item $C_{|\cdot|}(x,y^*_k)=y^*_k$, 
\item for $1\leq i\leq k$, $y_i$ is a prefix of $y_i^*$.
\end{itemize}

When there exists a polynomial $p\in \N[x]$ such that $|y^*_i|\leq p(|x|)$, for all $i\leq k$, the class is called $\DelKPolySpace$, \emph{$\calK$-delay with polynomial auxiliary memory}. When there exists a constant $c\in \N$ such that $|y^*_i|\leq |y_i|+c$, for all $i\leq k$, the class is called $\DelKConstantSpace$, \emph{$\calK$-delay with constant auxiliary memory}.
\end{definition}

The idea is that the $y^*_i$ will contain the previous solution plus the additional memory. Hence the superscript ``c'' indicates a bounded auxiliary memory size.

By abuse of expression, we will sometimes say that a problem in some of these classes above can be enumerated with a delay in $\calK$ or with a $\calK$-delay. When there is no restriction on memory i.e. when considering the class $\DelKPrecompSpace{T}$, an incremental  enumeration mechanism can be used.  Indeed, the memory can then store all solutions produced so far which results in an increase of the expressive power. 

Also in the case of memory, we allow possibly precomputation before the first output is made:

\begin{definition}\label{def:DelAC0PrecompSpace} {\rm [$\calK$-delay with $T$-precomputation and auxiliary memory] }

Let $R$ be a polynomially balanced predicate and $T$ be a complexity class. The enumeration problem $\enum{R}$ is in $\DelKPrecompSpace{T}$  if there exists an algorithm $M$ working with resource $T$ and two families of $\calK$-circuits $\calC=(C_n)_{n\in\N}$, $\calD=(D_n)_{n\in\N}$ such that, for all input $x$ there is an enumeration $y_1,...,y_k$ of $\sol{R}{x}$ and:

\begin{itemize}
\item $M$ computes some value $x^*$, i.e., $M(x)=x^*$,
\item $C_{|\cdot|}(x^*)=y^*_1$ and $D_{|\cdot|}(y^*_1)=y_1\in \sol{R}{x}$, 
\item for all $i<k$: $C_{|\cdot|}(x^*,y^*_i)=y^*_{i+1}$ and $D_{|\cdot|}(y^*_{i+1})=y_{i+1}\in \sol{R}{x}$, 
\item $C_{|\cdot|}(x^*,y^*_k)=y^*_k$,
\item for $1\leq i\leq k$, $y_i$ is a prefix of $y_i^*$.
\end{itemize}

When there exists a polynomial $p\in \N[x]$ such that $|y^*_i|\leq p(|x|)$, for all $i\leq k$, the class is called $\DelKPrecompPolySpace{T}$, \emph{$\calK$-delay with $T$-precomputation and polynomial auxiliary memory}. When there exists a constant $c\in \N$ such that $|y^*_i|\leq |y_i|+c$, for all $i\leq k$, the class is called $\DelKPrecompConstantSpace{T}$, \emph{$\calK$-delay with $T$-precomputation and constant auxiliary memory}.
\end{definition}

\subsection{Relation to Known Enumeration Classes}\label{subsec:relation to known classes}

All classes we consider in this paper are subclasses of the well-known classes $\DelayP$ or $\IncP$, resp., even if we allow our circuit to be of arbitrary depth (but polynomial size).

\begin{theorem}
  If $\calK, T\subseteq \ptime $, then 
$\DelKPrecompPolySpace{T}\subseteq \DelayP$ and
and $\DelKPrecompSpace{T}\subseteq \IncP$.
\end{theorem}

 

     

Let us briefly clarify the relation between our classes and the class $\cdl$ of enumeration problems that have a constant delay on a RAM after linear-time precomputation. This class was introduced in \cite{DurandG07}.

The problem to enumerate for a given graph the pairs of all vertices that are connected by a path of length 2 has only a polynomial number of solutions and is trivially in $\DelACzero$. 
 Since it is essentially the same as Boolean matrix multiplication, it is not in $\cdl$, assuming the beforementioned BMM hypothesis.
 
 On the other hand, note that the enumeration problem {\sc Enum-Parity}, given as input a sequence of bits with the solution set consisting only of one solution, the parity of the input, 
 is not in $\DelACzero$, since the parity function is not in $\ACzero$ \cite{FSS84,Ajtai89}. However, since Parity can be computed in linear time,  {\sc Enum-Parity} is trivially in $\cdl$.

As we will show in detail in the full version of this paper, the computation of a constant number of time steps of a RAM can be simulated by $\ACzero$ circuits. Hence if we add linear precomputation and polynomial memory to save the configuration of the RAM, we obtain an upper bound for $\cdl$. To summarize:

\begin{theorem}
Classes $\DelACzero$ and $\cdl$ are incomparable, and $\cdl\subsetneq\DelACzeroPrecompPolySpace{lin}$.
\end{theorem}

\section{Examples}
\label{examples}
In this section we show that many natural problems, ranging from graph problems, enumeration of Gray codes and satisfiability problems lie in our circuit classes.

\subsection{Graph Problems}

We first consider the enumeration problem associated with the notion of reachability in a graph. 
 
  \makeprob{\Ereach}{a graph $G=(V,E)$, $s\in V$}{an enumeration of vertices reachable from $s$}

\begin{theorem}
  $\pb{\Ereach}\in \DelACzeroPolySpace$
\end{theorem}

\begin{proof} 
At each step multiplication of Boolean matrices gives the set of vertices which are reachable from $s$ with one more step. This can be done in $\ACzero$.
The polynomial memory is  used   to remember all vertices that have been encountered to far.
\end{proof}

%
%
%

 Let us now turn to the enumeration of all transversals (not only the minimal ones).

  \makeprob{\Etransversal}{A hypergraph $H=(V,\calE)$}{an enumeration of all transversals of $H$}
  \smallskip

\begin{theorem}\label{lemma:Etransversal}
  $\Etransversal\in \DelACzero$.
\end{theorem}

\begin{proof}
 Let $\calE$ be a set of hyperedges over a set of $n$ vertices. Every  binary word $y=y_1\cdots y_n\in\{0,1\}^n$ can be interpreted as a subset of vertices.  We propose an algorithm that enumerates each of these words that corresponds to  a transversal of $H$ in lexicographical order with the convention $1<0$.  The algorithm is as follows:
\begin{itemize}
\item As a first step output $1\ldots 1$, the trivial solution. 
\item Let $H$  be the input and $y$ be the last output solution.
\begin{itemize} 
    \item For each   prefix $y_1\ldots y_i$ of $y$  with $y_i=1$ and $i\le n$ consider the word of length $n$, $z^i=y_1\ldots y_{i-1}01\ldots 1$.
    \item Check whether at least one of these words $z^i$  is a transversal of $H$.
        \item If yes select the one with the longest common prefix with $y$, that is the transversal  $z^i$ with the largest $i$ and output it as the next solution. 
        \item Else stop.
\end{itemize}
\end{itemize}

First we prove that the algorithm is correct. The transversal that is the successor of $y$ in our lexicographical order (where $1<0$), if it exists, has a common prefix with $y$, then a bit flipped from 1 to 0, and finally is completed by 1's only. Indeed, a successor of $y$ necessarily starts this way, and by monotonicity the first extension of such a prefix into a solution is the one completed by 1's only. As a consequence our algorithm explores all possible candidates and select the next transversal in the lexicographical order. 

Now let us prove that this is an $\ACzero$-delay enumeration algorithm that does not require  memory.
The main observation is that one can check with an $\ACzero$  circuit whether a binary word corresponds to transversal of $H$. Now, for each $i$ we can use a sub-circuit, which on input $(H$, $y$) checks whether $y_i=1$ and if yes whether $z^i$ is a transversal of $H$.  This  circuit can output $(z^i, 1)$ if both tests are positive, and $(y_i, 0)$ otherwise. All these sub-circuits can be wired  in parallel. Finally it suffices to use a selector to output $z^i$ with the largest $i$ for which $(z^i, 1)$  is output at the previous step. Such a selector can be implemented by  an $\ACzero$ circuit. 
\end{proof}

It is then easy to show (in a similar way) that  enumeration of all dominating sets of a graph can be done in $\DelACzero$.

%
  


\subsection{Gray Code}

Given $n\in \N$
, a Gray $n$-code is a ranked list of elements of $\Sigma^n$ such that between two successive words $x,y$ there exists only one bit such that $x_i\neq y_i$. Since we deal with Boolean circuits, we have to fix $\Sigma=\{0,1\}$, but Gray codes are defined for arbitrary alphabets.

The binary reflected Gray code of length $n$, denoted $G^n$, is made of $2^n$ words:
$G^n=[G^n_0,G^n_1, \ldots, G^n_{2^n-1}]$.
It is defined recursively as follows: $G^1=[0,1]$ and, for $n\geq 1$
\[G^n=[0G^{n-1}_0,0G^{n-1}_1, \ldots, 0G^{n-1}_{2^{n-1}-1}, 1G^{n-1}_{2^{n-1}-1}, \ldots, 1G^{n-1}_1, 1G^{n-1}_0].\]

As an example let us consider the list of pairs $(rank, word)$ for  $n=4$: $(0,0000)$,  $(1,0001)$, $(2,0011)$, $(3,0010)$, $(4,0110)$, $(5,0111)$, $(6,0101)$, $(7,0100)$, $(8,1100)$, $(9,1101)$, $(10,1111)$, $(11,1110)$, $(12,1010)$, $(13,1011)$, $(14,1001)$, $(15,1000)$.
%

Given $n$ and $r< 2^{n}$, let $b_{n-1}\cdots b_1b_0$ be the binary decomposition of $r$ and $G_r^n=a_{n-1}\cdots a_1a_0\in\Sigma^n$ be the $r$th word in the binary reflected code of length $n$. It is well-known that, for all $j=0,...,n-1$,
\[
b_j=\sum_{i=j}^{n-1} a_i\bmod 2 \text{ and } a_j = (b_j+b_{j+1})\bmod 2.
\]

Hence computing the rank of a word in the binary reflected code amounts to be able to compute parity. On the other side, computing the word from its rank can easily be done by a circuit. 


While it is trivial to enumerate all words of length $n$ in arbitrary or lexicographic order, this is not so clear for Gray code order. Also, given a rank or a first word, to enumerate all words of higher Gray code rank (in arbitrary order) are interesting computational problems.

\makeprob{\EGrayRank}{a binary word $r$ of length $n$ interpreted as an integer in $[0,2^n[$}{an enumeration of words of $G^n$ that are of rank at least $r$.}

\makeprob{\EGrayWord}{a word $x$ of length $n$}{an enumeration of words of $G^n$, that are of rank at least the rank of $x$.}

It turns out that for those problems where the order of solutions is not important, a very efficient enumeration is possible:

\begin{theorem}\label{th:gray code enumeration} Let $n$ be an integer
  \begin{enumerate} 
  \item Given $1^n$, enumerating all words of length $n$ even in lexicographic ordering is in $\DelACzero$
  \item $\EGrayRank\in \DelACzero$
  \item $\EGrayWord\in \DelACzero$
  \end{enumerate}
\end{theorem}

\begin{appendixproof}
\textbf{(of Theorem~\ref{th:gray code enumeration})}
The proof of the first item is immediate.

Let $r$ be an input of $\enum{\pb{Gray-Rank}}$. 
 From $r$ one can easily compute $y_0=G^n_r$ and $y_1=G^n_{r+1}$ but also $z_0=10\cdots01$ and $z_1=10\cdots00$ the two last words of the binary reflected Gray code. Suppose that $r$ is even (a similar argument can be given when $r$ is odd permuting the roles of $10\cdots01$ and $10\cdots00$. The enumeration step is the following:
\begin{itemize}
    \item As a first step, output $y_0$.
    \item Let $r$ be the input and $y$ be the last output solution. 
    \begin{itemize}
        \item Compute (again) $z_0$ and $y_1$
        \item If $y=y_1$, stop
        \item If $y=z_0$, then output $z_1$
        \item Else, switch the bit of $y$ at position $0$ then find the minimal position $i$ where there is a $1$ and switch bit at position $i+1$. Output this word as the new solution.
    \end{itemize}
\end{itemize}

%

The above process does not require memory.
Starting from $y_0$, it will start enumerating binary words of rank $G^n_{r+2},G^n_{r+4}, G^n_{r+6} ...$ until it reaches $z_0=10\cdots01=G^n_{2^n-2}$. It then outputs $z_1=10\cdots00=G^n_{2^n-1}$ and, applying the same rules, enumerates successively $G^n_{2^n-3}$, $G^n_{2^n-5}$, ... until $y_1=G^n_{r+1}$. 
The proof for  $\enum{\pb{Gray-Word}}$ proceeds along the same lines but is even simpler.
\end{appendixproof}

We next turn to those versions of the above problems, where we require that solutions are given one after the other in Gray code order. For each of them, the computational complexity is provably higher than in the above cases.

\begin{theorem}\label{lem:enumerating gray code words}
Given $1^n$, enumerating all words of length $n$ in a Gray code order  is in $\DelACzeroConstantSpace\backslash\DelACzeroPrecomp{\ptime}$
\end{theorem}

\begin{proof}
A classical method to enumerate gray code of length $n$ is the following~\cite{KreherStinson1999}. 
  \begin{itemize}
    \item Step $0$ : produce the word $0\cdots 0$ of length $n$.
    \item Step $2k+1$ : switch the bit at position $0$.
    \item Step $2k+2$: find minimal position $i$ where there is a $1$ and switch bit at position $i+1$. 
  \end{itemize}

  This method can be turned into an $\ACzero$-delay enumeration without precomputation using one bit of memory (to keep trace if the step is an even or odd one all along the computation). This proves the membership in  $\DelACzeroConstantSpace$.
  
  For the lower bound, suppose $\calC=(C_n)_{n\in\N}$ is an $\ACzero$ circuit family enumerating the Gray code of length $n$ after polynomial time precomputation produced by machine $M$. 
  We will describe how to use $\calC$ to construct an $\ACzero$-family computing the parity function, contradicting the lower bound given by \cite{Ajtai89,FSS84}. 
  
  Given is an arbitrary word $w=w_{n-1}\dots w_0$ of length $n$
  , and we want to compute its parity $\left(\sum_{i=0}^{n-1}w_i\right)\bmod2$. Let $x^*=M(1^n)$.
  Then, $w$ will appear as a solution somewhere in the enumeration defined by $\calC$. Let $w'$ be the next words after $w$. There exists $r$ such that $G^n_r=w$ and $G^n_{r+1}=w'$. By comparing $w$ and $w'$, one can decide which transformation step has been applied to $w$ to obtain $w'$ and thus if $r$ is odd or even. Note that the parity of $w$ is $1$ if and only if $r$ is odd. Hence, one can compute parity by a constant depth circuit operating as follows:
  \begin{tabbing}
  \qquad Inp\=ut $w$:\\
            \> $n:=|w|$;\\
            \> $x^* := M(1^n)$;\\
            \> $w' := C_{|\cdot|}(x^*,w)$;\\
            \> if last bits of $w$ and $w'$ differ then $v:=1$ else $v:=0$;\\
            \> output $v$.
  \end{tabbing}
  Note that the computation of $x^*$ does not depend on $w$ but only on the length of $w$; hence $x^*$ can be hardwired into the circuit family, which, since $M$ runs in polynomial time, will then be P-uniform. But we know from \cite{FSS84,Ajtai89} that parity cannot even be computed by non-uniform $\ACzero$ circuit families.
\end{proof}

We also consider the problem of enumerating all words starting not from the first one but at a given position, but now in Gray code order. Surprisingly this time the complexity will depend on how the starting point is given, by rank or by word.

\makeprob{\EGrayRankord}{A binary word $r$ of length $n$ interpreted as an integer in $[0,2^n[$}{an enumeration of words of $G^n$ in increasing number of ranks starting from rank $r$.}


\makeprob{\EGrayWordord}{A word $x$ of length $n$}{an enumeration of words of $G^n$ in Gray code order that are of rank at least  the rank of $x$.}

\begin{theorem}
\label{Lem: Gray code enumeration from rank}
\label{th: Gray code enumeration ordered bis} 
\begin{enumerate}
    \item $\EGrayRankord\in\DelACzeroConstantSpace\backslash\DelACzeroPrecomp{\ptime}$.
    \item $\EGrayWordord$ is  in the class $\DelACzeroPrecompConstantSpace{\ptime}$, but neither in $\DelACzeroPrecomp{\ptime}$ nor $\DelACzeroConstantSpace$.
\end{enumerate}
\end{theorem}
  
\begin{appendixproof}
\textbf{(of Theorem~\ref{th: Gray code enumeration ordered bis})}
\begin{enumerate}
\item
  For the upper bound,  we use the method described in Theorem~\ref{lem:enumerating gray code words}. Since the starting word is given by its rank $r$, one needs to modify a bit the approach above by first computing $G^n_r$, check what the parity of its last bit is to determine what kind of step needs to be performed first. Then we continue on as the above proof.

  
  Suppose $\enum{\pb{Gray-Rank}_{ord}}$ is in $\DelACzeroPrecomp{\ptime}$. Then by choosing $x=0^n$ we can enumerate all words of length $n$ in a Gray code order in $\DelACzeroPrecomp\ptime$, which, by Theorem~\ref{lem:enumerating gray code words}, is not possible.

\item
 For the membership in $\DelACzeroPrecompConstantSpace{\ptime}$ one just compute the rank $r$ of $x$ during the precomputation and use the preceding theorem.
 
 The first lower bound is proven exactly as in Theorem~\ref{Lem: Gray code enumeration from rank}.

The second lower bound follows by an easy modification:
Suppose $\enum{\pb{Gray-Word}_{ord}}\in\DelACzeroConstantSpace$. Given word $w$ of length $n$, we can compute its parity in $\ACzero$ as follows: Start the enumeration of $G^n$ to compute the first solution $w$ and next solution $w'$. Even with constant memory, this can be done by a circuit of constant-depth. Then decide the parity as in Theorem~\ref{lem:enumerating gray code words}.
\end{enumerate}
\end{appendixproof}

\subsection{Satisfiability Problems}
\label{sat}
 Deciding the satisfiability of a CNF-formula is  well-known to be $\NP$-complete. Nevertheless the problem becomes tractable for some restricted classes of formulas. For such classes we investigate the existence of an $\ACzero$-delay enumeration algorithm. First we consider monotone formulas.

\makeprob{\Emonotonesat}{A set of positive (resp. negative) clauses $\Gamma$ over a set of variables $V$}{an enumeration of all assignments over $V$ that satisfy $\Gamma$}

The following positive result is an immediate  corollary of  Theorem \ref{lemma:Etransversal}. 
  
\begin{theorem}\label{lemma:Emonotonesat}
 $\Emonotonesat \in \DelACzero$.
\end{theorem}

If we allow polynomial precomputation, then we obtain an $\ACzero$-delay enumeration algorithm for a class of CNF-formulas, referred to as  IHS in the literature (for Implicative Hitting Sets, see \cite{CreignouKS01}), which is larger than the monotone class.  A formula in this class consists of monotone clauses (either all positive or all negative) together with implicative  clauses.

  \makeprob{\Eihssat}{A set of clauses $\Gamma$    over a set of variables $V$, with $\calC=\calM\cup\calB$,
  where $\calM$ is a set of positive clauses (resp. negative clauses)  and $\calB$ a set of binary   clauses of the form $(\neg x)$ or $(x\lor\neg x')$ (resp. of of the form $( x)$ or $(x\lor\neg x')$)}
  {an enumeration of all assignments over $V$ that satisfy $\Gamma$}

\begin{theorem}\label{th:IHS}
   $\Eihssat\in \DelACzeroPrecomp{\ptime}\setminus\DelACzeroSpace$.
\end{theorem}
 \begin{proofsketch}
      Observe that contrary to the monotone case  $1....1$ is not a trivial solution.  Indeed a negative unary clause  $(\neg x)$ in $\calB$ forces $x$ to be assigned $0$, and this truth value can be propagated to other variables by the implicative clauses of the form $(x\lor \neg x')$.
     For this reason as a precomputation step, for each variable $x$ we compute $\tc(x)$ the set of  all variables that have to be set to $0$ in any assignment satisfying $\Gamma$ in which $x$ is assigned $0$. 
     With this information we can use an 
      algorithm that enumerates all truth assignments satisfying $\Gamma$ in lexicographical order  very similar to the one used for enumerating the transversals of a graph (see the proof of  Theorem \ref{lemma:Etransversal}).   
  The detailed algorithm can found in the appendix.
      
      For the lower bound,  consider the \textsc{st-connectivity} problem: given a directed graph $G=(V,A)$ with two distinguished vertices  $s$ and $t$,  decide whether there exists a path from $s$ to $t$. From $G$, $s$ and $t$ we build an instance of $\Eihssat$ as follows. We consider a set a clauses $\calC=\calP\cup\calB$,
  where $\calP=\{(s\lor t)\}$ and $\calB=\{(\neg s)\}\cup \{(x\lor \neg y)\mid (x,y)\in A\}$. This is an $\ACzero$-reduction.
  
  Observe that there exists a path from $s$ to $t$ if and only if $\Gamma$ is unsatisfiable. Suppose that $\Eihssat\in \DelACzero$, this means in particular that outputting a first assignment satisfying $\Gamma$ or deciding there is none is in $\ACzero$. Thus the above reduction shows that \textsc{st-connectivity} is in $\ACzero$, thus contradicting the fact that  \textsc{st-connectivity} is known  not to be in $\ACzero$ (see \cite{FSS84,Ajtai89}).
 \end{proofsketch}
  
  \begin{appendixproof} 
{\bfseries  (of Theorem \ref{th:IHS})}
  Let  $\Gamma$   be a set of clauses over a set of $n$ variables $V=\{x_1,\ldots x_n\}$, with $\Gamma=\calP\cup\calB$,
  where $\calP$ is a set of positive clauses    and $\calB$ a set of binary clauses of the form $(\neg x)$ or $(x\lor \neg x')$. Any truth assignment can be seen as a binary word of length $n$. 
     
     Observe that contrary to the monotone case  $1....1$ is not a trivial solution.  Indeed a negative unary clause  $(\neg x)$ in $\calB$ forces $x$ to be assigned $0$, and this truth value can be propagated to other variables by the implicative clauses of the form $(x\lor \neg x')$.
     For this reason as a precomputation step we propose the following procedure:
     \begin{itemize}
         \item Build a directed graph $G$  whose set of vertices is $V$. For any  2-clause $(x\lor \neg x')$  in $\calB$ there is an arc $(x,x')$ in $G$. 
         \item For each 
     variable $x$ compute $\tc(x)$ the set of vertices that are reachable from $x$ in $G$.
     \end{itemize} 
     
     Intuitively $\tc(x)$ contains all variables that have to be assigned  $0$  in any satisfying assignment in which   $x$ is assigned $0$.  
     Observe that any  variable $x$ such that  $(\neg x)\in \calB$ has to be assigned 0, and so have to be  all the variables in $\tc(x)$. We replace all these variables by their value and simplify the set of clauses accordingly.  If the empty clause occurs, then $\Gamma$ is not satisfiable, otherwise it is satisfiable by the $1\ldots 1$  assignment.
     
     So in the following w.l.o.g  we suppose that $\Gamma$ is satisfiable and has no negative unary variable.

The precomputation having been performed we  propose an algorithm that enumerates all truth assignments satisfying $\Gamma$ in lexicographical order with the convention $1<0$.  The algorithm is as follows:
\begin{itemize}
\item As a first step  output $1\ldots 1$ the trivial solution.
\item Let  the set of clauses, $\Gamma$, together with the set  of  lists of vertices reachable from each  vertex $x$ in the digraph $G$,   $\{\tc(x)| x\in V\}$, be the input and $y$ be the last output solution.
\begin{itemize} 
    \item For each   prefix $y_1\ldots y_i$ of $y$  with $y_i=1$ and $i\le n$ consider the word of length $n$, $z^i=y_1\ldots y_{i-1}0w_{i+1}\ldots w_n$,  where for $j\ge i+1$, $w_j=0$ if $\displaystyle x_j\in \tc(x_i) \cup \bigcup_{\{k\vert  k< i,  y_k=0\}}\tc(x_k)$, $w_j=1$ otherwise.
    \item Check whether at least one of these words $z^i$  is an assignment satisfying $\Gamma$.
        \item If yes select the one with the longest common prefix with $y$, that is the satisfying assignment  $z^i$ with the largest $i$ and output it as the next solution. 
        \item Else stop.
\end{itemize}
\end{itemize}

Observe that   $y_1\ldots y_{i-1}0$ is the prefix of a solution if and only if  $y_1\ldots y_{i-1}0w_{i+1}\ldots w_n$  as it is defined is a solution. Moreover if $y_1\ldots y_{i-1}0w_{i+1}\ldots w_n$ is a solution, then it is the first one in the considered lexicographic order with prefix $y_1\ldots y_{i-1}0$.   The proof that the algorithm is correct is then  similar to the one of Theorem \ref{lemma:Etransversal}. The precomputation can be done in polynomial time, and when done allows an implementation of the enumeration algorithm with constant-depth circuits, thus proving that $\Eihssat\in \DelACzeroPrecomp{\ptime}$.
\smallskip

For the lower bound,  consider the \textsc{st-connectivity} problem: given a directed graph $G=(V,A)$ with two distinguished vertices  $s$ and $t$,  decide whether there exists a path from $s$ to $t$. From $G$, $s$ and $t$ we build an instance of $\Eihssat$ as follows. We consider a set a clauses $\calC=\calP\cup\calB$,
  where $\calP=\{(s\lor t)\}$ and $\calB=\{(\neg s)\}\cup \{(x\lor \neg y)\mid (x,y)\in A\}$. This is an $\ACzero$-reduction.
  
  Observe that there exists a path from $s$ to $t$ if and only if $\Gamma$ is unsatisfiable. Suppose that $\Eihssat\in \DelACzero$, this means in particular that outputting a first assignment satisfying $\Gamma$ or deciding there is none is in $\ACzero$. Thus the above reduction shows that \textsc{st-connectivity} is in $\ACzero$, thus contradicting the fact that  \textsc{st-connectivity} is known  not to be in $\ACzero$ (see \cite{FSS84,Ajtai89}).
  \end{appendixproof}

  Surprisingly the enumeration method used so far for   satisfiability problems presenting a kind of monotonicity   can be used for the enumeration of all assignments satisfying a Krom set of clauses (i.e., a 2-CNF formula) as soon as the literals are considered in an appropriate order.  
  \begin{theorem}\label{th:KromSAT}
   $\EKromsat\in \DelACzeroPrecomp{\ptime}\setminus\DelACzeroSpace$.
\end{theorem}

\begin{proofsketch}
The proof  builds on the algorithm in~\cite{AspvallPT79} that decides whether a set of Krom clauses  is satisfiable in linear time. A full proof is given in the appendix.

Let $\Gamma$ be a set of $2$-clauses over a set of $n$ variables $V$. We perform the following precomputation steps:
\begin{itemize}
\item Build the associated implication graph, i.e., the     directed graph $G$  whose set of vertices is the set of literals  $V\cup \{\bar v: v\in V\}$. For any  $2$-clause $(l\lor l')$  in $\Gamma$ there are two arcs $\bar l\rightarrow l'$ and  $\bar l'\rightarrow l$ in $G$. 
\item For each literal $l$ compute $\tc(l)$ the set of vertices that are reachable from $l$ in $G$.
\item Compute the set of strongly connected components of $G$. If no contradiction is detected, that is if no strongly connected component contains both a variable $x$ and its negation, then contract each strongly connected component into one vertex. The result of this operation is a DAG, which, by abuse of notation, we also call $G$.
\item Compute a topological ordering of the vertices of $G$.
\item In searching through this  topological ordering,  build an ordered sequence $M$ of $n$ literals corresponding to the first occurrences of each variable.
\end{itemize} 
If the set of clauses is satisfiable, one can enumerate the satisfying assignments given as truth assignments on $M$ in   lexicographic order. The enumeration process is similar in spirit as the one developed in the preceding theorem. 

For the lower bound, the proof given in Theorem~\ref{th:IHS} applies.
\end{proofsketch}

\begin{appendixproof}
{\bfseries (of Theorem~\ref{th:KromSAT})}

The proof  builds on the algorithm in~\cite{AspvallPT79} that decides whether a set of Krom clauses  is satisfiable in linear time. 

Let $\Gamma$ be a set of $2$-clauses over a set of $n$ variables $V$. We perform the following precomputation steps:
\begin{itemize}
\item Build the associated implication graph, i.e., the  directed graph $G$  whose set of vertices is the set of literals   $V\cup \{\bar v: v\in V\}$. For any  $2$-clause $(l\lor l')$  in $\Gamma$ there are two arcs $\bar l\rightarrow l'$ and  $ \bar l'\rightarrow l$ in $G$. 
\item For each literal $l$ compute $\tc(l)$ the set of vertices that are reachable from $l$ in $G$.
\end{itemize}
Observe that $G$ has a  duality property, i.e., if $l\rightarrow l'$ is an arc of $G$, then so is $\bar l'\rightarrow \bar l$. Intuitively $\tc(l)$ contains all literals that have to be assigned 1 in any satisfying assignment in which $l$ is assigned 1.
Moreover,  a given  truth assignment is satisfying if and only   there is no arc $1\rightarrow 0$ in the graph in which every literal has been replaced by its truth value.

From this precomputation we can decide whether $\Gamma$ is satisfiable. Indeed, it is proven in \cite{AspvallPT79} that $\Gamma$ is satisfiable if and only if no strongly connected component of $G$ contains both a variable $x$ and its negation, i.e., there is no variable $x$ such that $x\in\tc(\bar x)$ and $\bar x\in\tc(x)$. We can also detect equivalent literals.

So in the following w.l.o.g  we suppose that $\Gamma$ is satisfiable and  has no equivalent literals. In particular $G$ is then a directed acyclic graph.

We then go on with  two additional precomputation steps:
\begin{itemize}
\item 
Compute a topological ordering of the vertices of $G$, which is denoted $\le$ in the following.
\item 
In searching through this   topological  ordering build an ordered sequence $M=(l_1, \ldots, l_n)$ of $n$ literals corresponding to the first occurrences of all variables i.e. for all $i,j\leq n$ s.t. $i\neq j$: $l_i\neq l_j$ and $l_i\neq \bar{l}_j$. 
\end{itemize}
\smallskip

The precomputation having been performed we  propose an algorithm that enumerates all truth assignments satisfying $\Gamma$, given as truth assignments on $M$, in lexicographic order.  The algorithm is as follows:
\smallskip
\begin{itemize}
\item As a first step  output $0\ldots 0$ the first solution.
\item Let  the set of clauses, $\Gamma$, together with the set  of  lists of vertices reachable from each  vertex $l$ in the implication graph $G$  be the input and $y$ be the last output solution.
\begin{itemize} 
    \item For each   prefix $y_1\ldots y_i$ of $y$  with $y_i=0$ and $i\le n$ consider the word of length $n$, $z^i=y_1\ldots y_{i-1}1w_{i+1}\ldots w_n$,  where for $j\ge i+1$, $w_j=1$ if $\displaystyle l_j\in \tc(l_i) \cup \bigcup_{\{k\vert  k< i,  y_k=1\}}\tc(l_k)$, $w_j=0$ otherwise.
    \item Check whether at least one of these words $z^i$  is an assignment satisfying $\Gamma$.
        \item If yes select the one with the longest common prefix with $y$, that is the satisfying assignment  $z^i$ with the largest $i$ and output it as the next solution. 
        \item Else stop.
\end{itemize}
\end{itemize}

To prove that the algorithm is correct we have to prove the following:
\begin{itemize}
    \item The assignment $l_1=0, \ldots, l_n= 0$ is satisfying.
    \item For all $i\le n-1$, $y_1\ldots y_i\in \{0,1\}^i$ is the prefix of a solution  if and only if 
    $y_1\ldots y_iw_{i+1}\ldots w_n$ 
     where for $j\ge i+1$, $w_j=1$ if $\displaystyle l_j\in \bigcup_{\{k\vert  k\le,  y_k=1\}}\tc(l_k)$, and  $w_j=0$ otherwise, 
    is a solution.
\end{itemize}
Observe that if $y_1\ldots y_iw_{i+1}\ldots w_n$ is a solution, then it is the first solution with prefix $y_1\ldots y_i$ in our lexicographic order. 
\smallskip

The fact that the assignment $l_1=0, \ldots, l_n= 0$ is satisfying follows from the proof in \cite{AspvallPT79}. In  seek of completeness let us reprove it.  In order to get a contradiction suppose it is not the case. This means that there are $l_i$ and $l_j$ in $M$ such that $\bar l_j\rightarrow l_i$. By duality one can suppose w.l.o.g that $i\le j$. On the one hand $i\le j$ implies that in the topological order $l_i\le l_j$, while on the other hand $\bar l_j\rightarrow l_i$ implies $\bar l_j\le l_i$. So we have $\bar l_j\le  l_i\le l_j$, which contradicts the fact  that $l_j$ (and not $\bar l_j$) is in $M$.
\smallskip

Now let us prove that $y_1\ldots y_i$    is the prefix of a solution  if and only if 
    $y_1\ldots y_iw_{i+1}\ldots w_n$ 
     where for $j\ge i+1$, $w_j=1$ if $\displaystyle l_j\in \bigcup_{\{k\vert  k\le i,  y_k=1\}}\tc(l_k)$, and  $w_j=0$ otherwise, 
    is a solution.
    
    One implication is trivial. So, let us suppose that $y_1\ldots y_iw_{i+1}\ldots w_n$ is not a solution. Then,   replacing the literals by their truth values in the graph makes appear an arc $1\rightarrow 0$. There is a discussion on the variables underlying this arc.
    
    If the observed contradiction involves two variables  whose truth values are fixed by  the prefix, then $y_1\ldots y_i$ is not the prefix of any solution.
    
    Suppose now that the observed contradiction involves two variables such that one has its truth value fixed by the prefix, the other not. Then  they are two literals $l_h$ and $l_j$ with $h\le i$ and $j \ge i+1$, such that by duality either:
    \begin{itemize}
        \item $l_h\rightarrow l_j$, $y_h=1$ and $w_j=0$, or 
        \item $l_h\rightarrow \bar l_j$, $y_h=1$ and $w_j=1$, or
        \item $\bar l_h\rightarrow l_j$, $y_h=0$ and $w_j=0$, or
        \item $\bar l_h\rightarrow \bar l_j$, $y_h=0$ and $w_j=1$.
    \end{itemize}
    
    
Suppose that $l_h\rightarrow l_j$, $y_h=1$ and $w_j=0$. The arc   $l_h\rightarrow l_j$ implies that    $l_j\in\tc(l_h)$, which  together with  $y_h=1$ implies $w_j=1$ by definition, a contradiction. 
  
  Suppose now that $l_h\rightarrow \bar l_j$, $y_h=1$ and $w_j=1$.
  On the one hand $l_h\rightarrow \bar l_j$ implies that any satisfying assignment that assigns 1 to  $l_h$, assigns 0 to $l_j$.
  On the other hand $w_j=1$ means   any satisfying assignment that starts by $y_1\ldots y_i$ assigns 1 to  $l_j$ .
  Hence, since  $y_h=1$  we get a contradiction.
  
  An arc $\bar l_h\rightarrow l_j$ cannot occur. Indeed, by duality we have then also $\bar l_j\rightarrow l_h$, which  implies that $\bar l_j \le l_h$ in the topological order. The fact that $h\le j$ implies $l_h\le l_j$. Thus we have $\bar l_j \le l_h\le l_j$. Hence contradicting the fact that $l_j$ (and not $\bar l_j$) is in $M$. 
  
  Finally  an arc $\bar l_h\rightarrow \bar l_j$ cannot occur either. Indeed by duality there is then also the arc $l_j\rightarrow l_h$, which implies $l_j\le l_h$ in the topological order. But, since $h\le j$, we have also $l_h\le l_j$, a contradiction. 
  \smallskip
  
  It remains to deal with the case where the observed contradiction involves two variables whose truth values are not fixed by the prefix
  Then  they are two literals $l_j$ and $l_k$ with $i+1\le j\le k$, such that by duality either:
    \begin{itemize}
        \item $l_j\rightarrow l_k$, $w_j=1$ and $w_k=0$, or 
        \item $l_j\rightarrow \bar l_k$, $w_j=1$ and $w_k=1$, or
        \item $\bar l_j\rightarrow l_k$, $w_j=0$ and $w_k=0$, or
        \item $\bar l_j\rightarrow \bar l_k$, $w_j=0$ and $w_k=1$.
    \end{itemize}

  Suppose $l_j\rightarrow l_k$, $w_j=1$ and $w_k=0$. By definition   $w_j=1$  means that there is an $h\le i$ such that $y_h=1$ and $l_j\in\tc(l_h)$. But then, the arc $l_j\rightarrow l_k$ implies that also $l_k\in\tc(l_h)$, thus contradicting the fact that $w_k=0$.
  
  Suppose now that $l_j\rightarrow \bar l_k$, $w_j=1$ and $w_k=1$.
  On the one hand $l_j\rightarrow \bar l_k$ implies that any satisfying assignment that assigns 1 to $l_j$   , assigns 0 to $l_k$. Since $w_j=1$ this means in particular that any satisfying assignment that starts by $y_1\ldots y_i$ assigns 0 to $l_k$.
  On the contrary $w_k=1$ means any satisfying assignment that starts by $y_1\ldots y_i$ assigns 1 to $l_k$, a contradiction.

    The last two cases cannot occur, for the same reasons as in the discussion above:  the existence of such arcs either contradicts the definition of $M$ or the definition of a topological order.

\smallskip

The precomputation can be done in polynomial time, and when done allows an implementation of the enumeration algorithm with constant-depth circuits, thus proving that $\EKromsat\in \DelACzeroPrecomp{\ptime}$.
\smallskip

For the lower bound, the proof given in Theorem~\ref{th:IHS} applies.
\end{appendixproof}

  We next turn to the special case where   clauses are XOR-clauses, i.e., clauses in which the usual ``or'' connective is replaced by  the exclusive-or connective,
  $\oplus$. Such a clause can be seen as a linear equation over the two elements field $\mathbb{F}_2$.

  \makeprob{\Exorsat}{A set of XOR-clauses $\Gamma$    over a set of variables $V$}
  {an enumeration of all assignments over $V$ that satisfy $\Gamma$}

If we allow  a polynomial precomputation step, then we obtain an $\ACzero$-delay enumeration algorithm for this problem that uses constant memory.  Interestingly this algorithm relies on the efficient enumeration of binary words in a Gray code order that we have seen in the previous section and contrary to the satisfiability problems studied so far does  not provide an enumeration in lexicographic order.  

\begin{theorem}\label{th:XOR}
   $\Exorsat\in \DelACzeroPrecompConstantSpace{\ptime}\setminus \DelACzeroSpace$.
\end{theorem}
 
\begin{proofsketch}
Observe that a set of XOR-clauses  $\Gamma$ over a set of variables $V=\{x_1,\ldots x_n\}$ can be seen as a linear system over $V$ on the two elements field $\mathbb{F}_2$.
As a consequence enumerating  all assignments over $V$ that satisfy $\Gamma$ comes down to enumerating all solutions of the corresponding linear system.

As a precomputation step  we apply Gaussian elimination in order to obtain an equivalent triangular system. If the system has no solution stop.
    Otherwise we can suppose that the linear system is of rank $n-k$ for some $0\le k\le n-1$,  and without loss of generality that $x_1,\ldots, x_k$ are free variables, whose assignment determines the assignment of all other variables in the triangular system. 
    We then compute a first solution $s_0$ corresponding to $x_1,\ldots, x_k$ assigned  $0\ldots 0$. 
    Next, for each $i=1, \ldots, k$ compute the  solution $s_i$ corresponding to all variables in  $x_1, ,\ldots x_k$ assigned  0 except $x_i$ which is  assigned  1. Compute then the influence list of $x_i$, 
    $L(x_i)=\{j \mid k+1\le j\le n,  s_0(x_j)\ne s_i(x_j)\}$. The influence list of $x_i$ gives the bits that will be changed when going from a solution to another one in flipping only the bit $x_i$ in the prefix corresponding to the free variables. 
    Observe that this list does not depend on the solution ($s_0$ in the definition) we start from.
 
With this precomputation  we start our enumeration procedure, which uses the enumeration of binary prefixes of length $k$ in a Gray code order as a subprocedure. 
The algorithm is described in the appendix.
\end{proofsketch}

\begin{appendixproof} 
{\bfseries (of Theorem \ref{th:XOR})}
Observe that a set of XOR-clauses  $\Gamma$ over a set of variables $V=\{x_1,\ldots x_n\}$ can be seen as a linear system over $V$ on the two elements field $\mathbb{F}_2$.
As a consequence enumerating  all assignments over $V$ that satisfy $\Gamma$ comes down to enumerating all solutions of the corresponding linear system.
As a precomputation step  we propose the following procedure:
\begin{itemize}
    \item Apply Gaussian elimination in order to obtain an equivalent triangular system. If the system has no solution stop.
    Otherwise we can suppose that the linear system is of rank $n-k$ for some $0\le k\le n-1$,  and without loss of generality that $x_1,\ldots, x_k$ are free variables, whose assignment determines the assignment of all other variables in the triangular system. 
    \item Compute a first solution $s_0$ corresponding to $x_1,\ldots, x_k$ assigned  $0\ldots 0$. 
    \item For each $i=1, \ldots, k$ compute the  solution $s_i$ corresponding to all variables in  $x_1, ,\ldots x_k$ assigned  0 except $x_i$ which is  assigned  1. Compute then the influence list of $x_i$, 
    $L(x_i)=\{j \mid k+1\le j\le n,  s_0(x_j)\ne s_i(x_j)\}$. The influence list of $x_i$ gives the bits that will be changed when going from a solution to another one in flipping only the bit $x_i$ in the prefix corresponding to the free variables. 
    Observe that this list does not depend on the solution ($s_0$ in the definition) we start from.
\end{itemize}

With this precomputation  one can easily output a first solution, corresponding to the prefix $x_1=\ldots =x_k=0$, and then start our enumeration procedure, which uses the enumeration of binary prefixes of length $k$ in a Gray code order as a subprocedure. 
The algorithm is as follows:
\begin{itemize}
\item As a first step compute a first solution $s_0$ corresponding to $x_1,\ldots, x_k$ assigned  $0\ldots 0$. If it exists, output it, else stop. 
\item Let the triangular system obtained from Gaussian elimination  together with the set of influence lists of all free variables  $\{L(x_i)| 1\le i\le k\}$ be the input and $s=ww_{k+1}\ldots w_n$, where $w$ is a prefix of length $k$, be the last output solution.
\begin{itemize} 
    \item Compute $w'$ the successor of $w$ in a Gray code order enumeration.
    \item If it exists, then $w'$ and $w$ differ only on one position, say the $i$th. Compute $s'=w'w'_{k+1}\ldots w'_n$ where for $j\ge k+1$, $w'_j=1-w_j$ if $j\in L(x_i)$, and $w'_j=w_j$ otherwise. Output $s'$ as the next solution. 
        \item Else stop.
\end{itemize}
\end{itemize}

The precomputation runs in polynomial time. The system has $2^k$ solutions, one for each possible prefix on $x_1,\ldots, x_k$. According to Theorem \ref{lem:enumerating gray code words} the enumeration of these prefixes in a Gray code order can be done with $\ACzero$-delay in using a constant space. 
Two successive words differ exactly on one index $i$. We can then go from one solution to the next one in flipping in the former solution the variables in the influence list of $x_i$. Since the influence lists have been computed in the precomputation step, this can be done done with constant depth circuits with no additional memory, thus concluding the proof.

To show that $\Exorsat\not\in \DelACzeroSpace$ we remark that a word $w=w_1...w_n\in\{0,1\}^n$ has an even number of ones iff the following set of XOR-clauses $\Gamma$ has a solution:
\[
\Gamma=\{x_1=w_1,x_2=w_2,...,x_n=w_n, x_1\oplus x_2\oplus \cdots\oplus x_n=0\}.\]
  
Moreover, all words of length $n$ can be mapped to such systems of the same length. 
\end{appendixproof}

\section{Separations of Delay Classes}
\label{separations}

In the previous results we already presented a few lower bounds, but now we will systematically strive to separate the studied classes.  

As long as no precomputation is allowed, we are able to separate all delay classes (only the case of unbounded memory, so the ``ìncremental class'' resists). With precomputation, the situation seems to be more complicated. We obtain only a conditional separation of the class with constant memory from the one without memory at all.

\subsection{Unconditional Separations for Classes without Precomputation}

\begin{theorem}\label{th:nospacestriclyinconstspace}
  $\DelACzero \subsetneq \DelACzeroConstantSpace$
\end{theorem}

\begin{proof}
  Let $x\in \{0,1\}^*$, $|x|=n\in \N^*$, $x=x_1\dots x_n$. We denote by $m=\lceil \log n\rceil+1$.
  Let $R_L$ be defined for all $x\in \{0,1\}^*$ as the union of the two following sets $A$ and $B$:
  
\begin{itemize}
  \item $A=\bigl\{\,y\in \{0,1\}^* \big| |y|=m, y\neq 0^{m} , y\neq 1^{m} \bigr\}$
  \item $B=\{\,1^{m}\}$  if $x$ has an even number of ones, else $B=\{\,0^{m}\}$.
\end{itemize}
  
We denote by $z_1,...,z_t$ an enumeration of elements of $A$. Clearly, $|R_L(x)|=t+1$ and $t\geq n$. 
To show that $R_L\in \DelACzeroConstantSpace$, we use the enumeration of elements of $A$ (which is easy) and one additional memory bit
that is transferred from one step to the other to compute $\pb{Parity}$. Indeed, we build families of circuits $(C_n)$ and $(D_n)$ according to Definition~\ref{def:DelAC0Space} as follows.

  \begin{itemize}
    \item First $C_{|\cdot|}(x)$ computes $y_1^*=z_1b^1$ where $b^1=x_1$, and $D_{|\cdot|}(y_1^*)=z_1$. 
    \item For $1<i\leq t$, the circuit $C_{|\cdot|}(x,y_{i-1}^*)$ computes $y_{i}^*$, where $y^*_{i}=z_{i}b^{i}$ with $b^{i}=b^{i-1}\oplus x_{i}$ if $i\leq n$, and $b^{i}=b^{i-1}$ else,
    and $D_{|\cdot|}(y_{i}^*)=z_{i}$.
    \item After $t$ steps, the memory bit $b^t$ contains a $0$ if and only if the number of ones in $x$ is even. According to this, we either output 
    $1^{m}$ or $0^m$ as last solution.
  \end{itemize}

  Note that the size of the solutions is $m$, the size of the memory words above is $m+1$, hence we need constant amount of additional memory. The circuit families $(C_n)$ and $(D_n)$ are obviously DLOGTIME-uniform.

  Suppose now that $R_L\in \DelACzero$ and let $(C_n)$ be the associated family of enumeration circuits. 
  We construct a circuit family as follows:
  We compute in parallel all $C_{|\cdot|}(x)$ and $C_{|\cdot|}(x,z_i)$ for $1\leq i\leq t$. In this way, we will obtain among other solutions either $0^m$ or $1^m$. We accept in the first case. Note that the $z_i$ are the same for all inputs $x$ of the same length. Thus, we obtain an $\ACzero$ circuit family for parity, contradicting \cite{FSS84,Ajtai89}.
\end{proof}

By extending the above approach, one can prove the following separation:

\begin{theorem}\label{th:constant space strictly in poly space}
  $\DelACzeroConstantSpace \subsetneq \DelACzeroPolySpace$
\end{theorem}

\begin{appendixproof}
\textbf{(of Theorem~\ref{th:constant space strictly in poly space})}.
  For a given $x\in \{0,1\}^*$, $n=|x|$, we denote by $m=\lceil \log\log n\rceil+1$.
  Let $R_L$ be defined for all $x\in \{0,1\}^*$ as the union of the two following sets $A$ and $B$:
 
\begin{itemize}
  \item $A=\bigl\{\,y\in \{0,1\}^* \big| |y|=m, y\neq 0^{m} , y\neq 1^{m} \bigr\}$
  \item $B=\{\,1^{m}\}$  if $x$ has an even number of ones, else $B=\{\,0^{m}\}$.
\end{itemize}
  
We denote by $z_1,...,z_t$ an enumeration of elements of $A$. Clearly, $|R_L(x)|=t+1$ and $t\geq \log n$. 

To show that $R_L\in \DelACzeroPolySpace$, we mimic the proof of the computation of the parity of $n$ input bits with a tree of binary parity-gates.
Indeed, we build families of circuits $(C_n)$ and $(D_n)$ according to Definition~\ref{def:DelAC0Space} as follows.

  \begin{itemize}
    \item First $C_{|\cdot|}(x)$ computes $y_1^*=z_1b^1$ where $b^1_j=x_{2j-1}\oplus x_{2j}$; if $n$ is odd then $b^1_{\lceil n/2\rceil}=x_n$. We set $D_{|\cdot|}(y_1^*)=z_1$. 
    \item For $1<i\leq t$, the circuit $C_{|\cdot|}(x,y_{i-1}^*)=y_{i}^*$, where $y^*_{i}=z_{i}b^{i}$ where $b^i_j=b^{i-1}_{2j-1}\oplus b^{i-1}_{2j}$ for $1\leq j\leq \lceil n/2\rceil$; if $|b^{i}|$ is odd then the last bit of $b^i$ is the last bit of $b^{i-1}_n$. 
    We set $D_{|\cdot|}(y_{i}^*)=z_{i}$.
    \item After $t$ steps, the bit $b^t_1$ contains a $0$ if and only if the number of ones in $x$ is even. According to this, we either output 
    $1^{m}$ or $0^m$ as last solution.
  \end{itemize}


  Note that the memory words above are of linear size.

 Suppose now that $R_L\in \DelACzeroConstantSpace$. Let $(C_n)$, $(D_n)$ be the associated families of enumeration circuits and let $c\in \N$ be the space allowed for additional memory. 
 Fix an input $x$, $|x|=n$. Let $t$ be as above.
 Now for any sequence $b^1,\dots,b^{t+1}$ of additional memory, where one $b^i$ is empty and all the others have length exactly $c$, we can check in $\ACzero$ that this is a correct sequence, which means that for every $b^i$, $C_{|\cdot|}(x,z_ib^i)$ is of the form $z_jb^j$ for some $j$. Here, for the empty memory word $\varepsilon$ we let $z\varepsilon=\varepsilon$. For the correct sequence, we check as in the proof of Theorem~\ref{th:nospacestriclyinconstspace} if the solution $0^m$ or $1^m$ will appear.
 
 Since the length of a sequence of additional memory words is at most $ct$, their number is polynomial in $n$. Hence we obtain an $\ACzero$ circuit for parity which does not exist.
\end{appendixproof}

The \textsc{parity} problem can be seen as an enumeration problem: given $x$, one output the unique solution $1$ if the number of ones in $x$ is even. One outputs $0$ if it is odd. Since as a function problem, \textsc{parity} can not be in $\DelACzeroPolySpace$ (the fact there is only one solution makes memory useless). It is obviously in $\DelayP$. This implies that $\DelACzeroPolySpace\subsetneq \DelayP$. Putting all the previous results together, we conclude:

\begin{corollary}
 $\DelACzero\subsetneq \DelACzeroConstantSpace \subsetneq \DelACzeroPolySpace\subsetneq \DelayP$. 
\end{corollary}

\subsection{Conditional Separation for Classes with Precomputation}
\label{separations:conditional}


If precomputation is allowed, the separation proofs of the previous subsection no longer work; in fact we do not know if the corresponding separations hold. However, under reasonable complexity-theoretic assumptions we can at least separate the classes $\DelACzeroPrecomp\ptime$ and $\DelACzeroPrecompConstantSpace{\ptime}$. Note that in Theorem~\ref{lem:enumerating gray code words} we already proved a separation of just these two classes, but this concerns only the special case of ordered enumeration, and does not say anything about the general case.
We find it interesting that the proof below relies on a characterization of the class $\PSPACE$ in terms of regular leaf-languages or \emph{serializable computation} \cite{helascvowa93,Vollmer99a}. 
The proof will be given in the appendix.

\begin{theorem}\label{thm:condsep}
If $\NP\neq \PSPACE$, then $\DelACzeroConstantSpace \setminus \DelACzeroPrecomp{\ptime}\ne\emptyset$.
\end{theorem}
 
\begin{appendixproof} 
{\bfseries (of Theorem~\ref{thm:condsep})}
  In 1993 Hertrampf \textit{et al.} \cite{helascvowa93} proved that $\PSPACE$ is $\ACzero$-serialisable, i.\,e., every $\PSPACE$ decision algorithm can be divided into an exponential number of slices, each requiring only the power of $\ACzero$ and passing only a constant number of bits to the next slice.
More formally,  for every language $L\in\PSPACE$ there is an $\ACzero$-circuit family $\calC=(C_n)_{n\in\N}$, numbers $k,l\in\N$ such that every $C_n$ has $k$ output bits and for every input $x$, $\vert x\vert=n$, $N=n+n^l+k$, $x\in L$ if and only if $c_{2^{n^l}}=\underbrace{1\ldots 1}_k$ where
  \begin{itemize}
      \item $y_i$ is the $i$-th string in $\{0,1\}^{n^l}$ in lexicographic order,
      \item $c_1=C_N(x, y_1,\underbrace{0\ldots 0)}_k$, and 
      \item $c_i=C_N(x,y_i, c_{i-1})$ for $1<i\le 2^{n^l}$.
  \end{itemize}
  
  Depending on $L$ we now define the relation 
  $$R_L(x)=\bigl\{\,0y\in\{0,1\}^* \,\big|\, \vert y\vert =n^l\,\bigr\}\cup \{\,1^{n^{l}+1}\mid x\in L\,\}\cup \{\,10^{n^{l}}\mid x\not\in L\,\}.$$
  To enumerate $R_L$ with constant auxiliary memory, we construct circuit families $\calC'=(C'_n)_{n\in\N}$ and $\calD'=(D'_n)_{n\in\N}$ as follows:
  \begin{itemize}
      \item $C'_{|\cdot|}(x) = (y_0,c)$ where $c=C_N(x,y_0,0\dots0)$
      \item $C'_{|\cdot|}(x,(y_i,c))=(y_{i+1},c')$ for $0\leq i<2^{n^l}$ and $c'=C_N(x,y,c)$,
      \item $C'_{|\cdot|}(x,(y_{2^{n^l}},c))=(1^{n^{l}+1},c)$
      \item $C'_{|\cdot|}(x,(1^{n^{l}+1},c))=(1^{n^{l}+1},c)$
      \item $D'_{|\cdot|}(y_i,c)=0y_i$ for $1\leq i\leq2^{n^l}$
      \item $D'_{|\cdot|}(1^{n^{l}+1},c)= 1^{n^{l}+1}$ if $C_N(x,y_{2^{n^l}},c)= \underbrace{1\ldots 1}_k$, $10^{n^{l}}$ else.
  \end{itemize}
      
      Thus we see that $R_L$ can be enumerated with $k$ bits of auxiliary memory, i.e, 
      $R_L\in\DelACzeroConstantSpace$.
      

Now suppose  $\enum R_L\in\DelACzeroPrecomp{\ptime}$ via $\ACzero$-circuit family $\calC=(C_n)_{n\in\N}$. Then $L\in \NP$ by the following algorithm.
\begin{itemize}
 \item Given $x$, use precomputation to produce $x^*$.
 \item Check if $C_{|\cdot|}(x^*)=1^{n^{k}+1}$. If yes, accept.
 \item Guess some output $y\in R_L(x)$.
 \item Check that $C_{|\cdot|}(x^*,y)=1^{n^{k}+1}$.
 \item If yes, then accept.
\end{itemize}
\end{appendixproof}

%
%
%

\begin{figure}[!ht]
  \centering
\begin{tikzpicture}[scale=0.5]
\node[draw](a) at (0,0){$\DelACzero$};
\node[draw](b) at (0,3) {$\DelACzeroConstantSpace$};
\node[draw](c) at (0,6) {$\DelACzeroPolySpace$};
\node[draw](d) at (0,9) {$\DelACzeroSpace$};
\node[draw](e) at (5,2) {$\DelACzeroPrecomp{\mathrm{P}}$};
\node[draw](f) at (5,5) {$\DelACzeroPrecompConstantSpace{\mathrm{P}}$};
\node[draw](g) at (5,8) {$\DelACzeroPrecompPolySpace{\mathrm{P}}$};
\node[draw](h) at (5,11) {$\DelACzeroPrecompSpace{\mathrm{P}}$};
%
\node[draw](j) at (10,10) {$\DelayP$};
\node[draw](k) at (10,13) {$\IncP$};

\draw[line width=1pt](a) to (b);
\draw[line width=1pt](b) to (c);
\draw[dashed](c) to (d);
\draw[line width=1pt](e) to  node[midway,right]{(if $\NP\ne\PSPACE$)}(f);
\draw[dashed](f) to (g);
\draw[dashed](g) to (h);
%
\draw[line width=1pt](a) to (e);
\draw[line width=1pt](b) to (f);
\draw[line width=1pt](c) to (g);
\draw[line width=1pt](d) to (h);
%
\draw[dashed](j) to (k);
\draw[dashed](j) to (g);
\draw[dashed](h) to (k);

\end{tikzpicture}
\caption{Diagram of the classes. Bold lines denote strict inclusions.}\label{graphe1}
\end{figure}
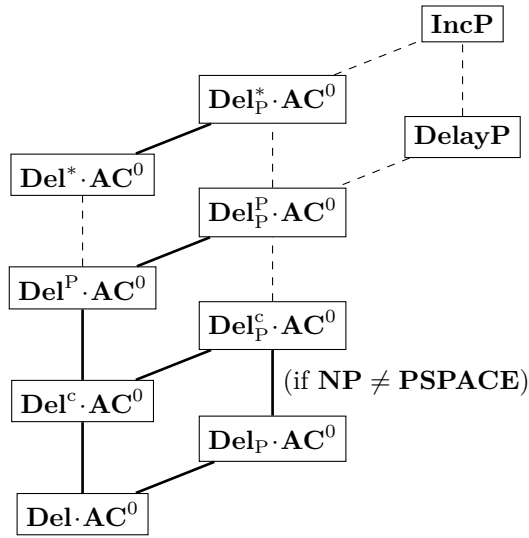

\section{Conclusion}

The obtained inclusion relations among the classes we introduced are summarized in Fig.~\ref{graphe1}. We noted earlier that in our context, enumeration problems are defined without a complexity assumption concerning the underlying relation. We should remark that quite often, a polynomial-time upper bound is required, see \cite{Str19,CapelliS19}. All of our results, with the exception of the conditional separations in Sect.~\ref{separations} also hold under the stricter definition; however, the relation $R_L$ used in the lower bounds in Subsect.~\ref{separations:conditional} is based on a $\PSPACE$-complete set and therefore, to check if $y\in R_L(x)$ requires polynomial space w.r.t.~the length of $x$. It would be nice to be able to base these separations on polynomial-time checkable relations, or even better, to separate the classes unconditionally, but this remains open. Moreover, some further inclusions in Fig.~\ref{graphe1} are still not known to be strict.

In Subsect.~\ref{sat}, we proved that, for several fragments of propositional logic, among them the Krom and the affine fragments, the enumeration of satisfiable assignments is in the class $\DelACzeroPrecomp{\ptime}$. This means satisfiable assignments can be enumerated very efficiently, i.\,e., by an $\ACzero$-circuit family, after some precomputation, which is also efficiently doable (in polynomial time). 
For another important and very natural fragment of propositional logic, namely the Horn fragment, a $\DelayP$-algorithm is known, but it is not at all clear how polynomial-time precomputation can be of any help to produce more than one solution.
Since {\sc Horn-Sat} is $\ptime$-complete, we conclude that $\EHornsat\not\in\DelACzeroSpace$, and we conjecture that it is not in $\DelACzeroPrecomp\ptime$. In fact, we do not see any reasonable better bound than the known $\DelayP$. 

\newpage

\bibliography{biblio}

\end{document}